\documentclass {imsart}
\usepackage{url}
\usepackage{natbib}
\usepackage{graphicx}
\usepackage{color}
\usepackage{xspace}
\usepackage{amsthm,amsmath,natbib, amsfonts}

\newtheorem{Definition}{Definition}
\newtheorem{theorem}{Theorem}
\newtheorem{Corollary}{Corollary}
\newtheorem{Proposition}{Proposition}
\newtheorem{lemma}{Lemma}

\renewcommand{\emph}[1]{\textsl{#1}}
\newcommand{\SUD}{$U(0,1)$\xspace}
\newcommand{\var}{\mathrm{var}}
\newcommand{\FWER}{\text{$\mathrm{FWER}$}}
\newcommand{\FWERcb}{\mathrm{FWER}_{\mathrm{CBP}}^{\lambda,\alpha}}
\newcommand{\FDR}{\mathrm{FDR}}
\renewcommand{\P}{\mathcal{P}}
\newcommand{\Plam}{\mathcal{P}^\lambda}
\newcommand{\plim}{\mathop{\mathrm{plim}}}
\newcommand{\cor}{\mathrm{cor}}
\newcommand{\E}{\mathrm{E}}
\newcommand{\R}{\mathbb{R}}                 
\newcommand{\diag}{\mathrm{diag}}
\def\mbf#1{\boldsymbol{#1}}                

\begin{document}
\begin{frontmatter}

\title{Gaining power in multiple testing of interval hypotheses via conditionalization}
\runtitle{Interval hypotheses}

\author{\fnms{Jules L.} \snm{Ellis$^1$} \corref{} } 
\and
\author{\fnms{Jakub} \snm{Pecanka$^2$} } 
\and
\author{\fnms{Jelle} \snm{Goeman$^2$} }

\affiliation{Radboud University Nijmegen \thanksmark{m1} and Leiden University Medical Center \thanksmark{m2}}


%

\runauthor{Ellis et al.}

\address{1. Radboud University Nijmegen, 2. Leiden University Medical Center}

\begin{abstract}
In this paper we introduce a novel procedure for improving multiple testing procedures (MTPs) under scenarios when the null hypothesis $p$-values tend to be stochastically larger than standard uniform (referred to as \emph{inflated}). An important class of problems for which this occurs are tests of interval hypotheses. The new procedure starts with a set of $p$-values and discards those with values above a certain pre-selected threshold while the rest are corrected (scaled-up) by the value of the threshold. Subsequently, a chosen family-wise error rate (FWER) or false discovery rate (FDR) MTP is applied to the set of corrected $p$-values only. We prove the general validity of this procedure under independence of $p$-values, and for the special case of the Bonferroni method we formulate several sufficient conditions for the control of the FWER. It is demonstrated that this 'filtering' of $p$-values can yield considerable gains of power under scenarios with inflated null hypotheses $p$-values.
\end{abstract}



\end{frontmatter}

\section{Introduction}

Multiple testing procedures (MTPs) generally assume that $p$-values of true null hypotheses follow the standard uniform distribution or are stochastically larger. The latter situation may occur when interval (rather than point) null hypotheses are tested. Under such scenarios the $p$-values are standard uniform typically only in borderline cases such as when the true value of a parameter is on the edge of the null hypothesis interval. When the true value of the parameter is in the interior of the interval the $p$-values tend to be stochastically larger than uniform, sometimes dramatically so with many $p$-values having distribution concentrated near 1. We call such $p$-values \emph{inflated}.

There are many practical examples of multiple testing situations with interval null hypotheses with some or all of the true parameter values located (deep) in the interior of the null hypothesis interval. (1.) In test construction according to the nonparametric Item Response Theory (IRT) model of Mokken (1971), one can test whether all item-covariances are nonnegative \citep{Mokken1971, Rosenbaum1984, HR1986, JE1997}. Ordinarily, most item-covariances are substantially greater than zero, with only a few negative exceptions.
(2.) In large-scale survey evaluations of public organisations, such as schools or health care organizations, it can be interesting to test whether organizations score lower than a benchmark \citep{Normand2007, Ellis2013}. If many organisations score well above the benchmark, a large number of the $p$-values of true null hypotheses become inflated. (3.) When a treatment or a drug is known to have a substantial positive treatment effect in a given population it can be of interest to look for adverse treatment effects in subpopulations. The $p$-values of most null hypotheses again become inflated.

Intuitively, if the null $p$-values tend to be stochastically larger than uniform, true and false null hypotheses should be easier to distinguish, making the multiple testing problem easier. However, most MTPs focus the error control on the 'worst case' of standard uniformity and thus miss the opportunity to yield more power for inflated $p$-values. Consequently, in the presence of inflated $p$-values the actual error rate can be (much) smaller than the nominal level. Some MTPs actually lose power when null $p$-values become inflated, e.g.\ adaptive FDR methods \citep[e.g.][]{Storey2002} that incorporate an estimate of $\pi_0$, the proportion of null $p$-values \citep[note 9 on p.\ 258]{Fischer2012}.

In this paper we propose a procedure which improves existing MTPs in the presence of inflated $p$-values by adding a simple conditionalization step at the onset of the analysis. For an a priori selected threshold $\lambda\in(0,1]$ (e.g.\ $\lambda=0.5$) we remove (i.e. not reject) all hypotheses with $p$-value above $\lambda$. The remaining $p$-values are scaled up by the threshold $\lambda$: $p_i'=p_i/\lambda $. The selected MTP is subsequently performed on the rescaled $p$-values only. We refer to the altered procedure as the \emph{conditionalized} version of the MTP, leading to procedures such as the \emph{conditionalized Bonferroni procedure} (CBP).

In terms of power, there are both benefits and costs associated with conditionalization. The costs come from the scaling of the $p$-values by $1/\lambda$, thus effectively increasing their values. If a fixed significance threshold were used, the number of significant $p$-values would decrease. However, the conditionalization step also tends to increase the significance threshold for each $p$-value by reducing the multiple testing burden (i.e. the number of hypotheses corrected for). Crucially, in scenarios with a large portion of substantially inflated $p$-values the increased significance threshold means that the overall effect of conditionalization results in a more powerful procedure.


In the remainder of the paper we formally investigate the effects and benefits of conditionalization. We prove that for scenarios with inflated $p$-values conditionalized procedures retain type I error control whenever $p$-values are independent. Our result applies to both the family-wise error rate (FWER), the false discovery rate (FDR), and other error rates. We also show that if $p$-values are not independent, such control is not automatically guaranteed. We formulate conditions which are sufficient for the control of the FWER by the CBP. We conjecture that the CBP is generally valid for positively correlated $p$-values. Finally, the power of conditionalized procedures is investigated using simulations.

\section{Definition of conditionalized tests}

We define a multiple testing procedure (MTP) $\P$ as a mapping that transforms any finite vector of $p$-values into an equally long vector of binary decisions. If $\P({p_1},\ldots,{p_m})=({d_1},\ldots,{d_m})$, then ${d_i}$ indicates whether the null hypothesis corresponding to $p_i$ is rejected ($d_i=1$) or not ($d_i=0$). We define a decision rate as the expected value of a function of $\P({p_1},\ldots,{p_m})$. We denote the FWER and FDR of the procedure $\P$ as $\FWER_{\P}$ and $\FDR_{\P}$, respectively. 


For $\lambda\in(0,1]$ and an MTP $\P$ we define the corresponding conditionalized MTP $\Plam$ as the MTP that, on input of a vector of $p$-values $({p_1},\ldots,{p_m})$,  applies $\P$ to the sub-vector consisting of only the rescaled $p$-values $p_i/\lambda$ with $p_i\leq\lambda$, and that does not reject the null hypotheses of the $p$-values with $p_i > \lambda$. Throughout the paper we always assume that both the level of significance $\alpha$ and the conditionalization factor $\lambda$ are fixed (independently of the data) prior to the analysis.





In this paper we pay special attention to the conditionalized Bonferroni procedure (CBP) and its control of the FWER. For $\lambda\in(0,1]$ define $R_m(\lambda)=\sum_{i=1}^m\boldsymbol{1}\{p_i\le\lambda\}$. Let $\mathcal{T}\subseteq\{1,\ldots,m\}$ be the index set of true null hypotheses. The FWER of the CBP is defined as
\begin{align*}
\FWER_{\text{CB}}^{\lambda,\alpha}=P\Big(\bigcup\limits_{i\in\mathcal{T}}{\Big[\,p_i<\frac{\alpha \lambda }{R_m(\lambda) \vee 1}}\Big]\Big).
\end{align*}
If $\FWER_{\text{CB}}^{\lambda,\alpha}\le\alpha$ for given $\lambda$ and $\alpha$ we say that \emph{the CBP controls the FWER}
for those $\lambda$ and $\alpha$.

Note that for the sake of simplicity in the rest of the paper we sometimes suppress one or both arguments and simply use $R_m$, $R(\lambda)$, or even $R$ in the place of $R_m(\lambda)$. The proofs of all theorems and lemmas formulated below can be found in the Appendix.

\section{FWER and FDR of independent tests}

In this section we state our main result: a conditionalized procedure controls FWER (or FDR) if the non-conditionalized procedure controls FWER (or FDR) and if the test statistics are independent and the marginal distributions satisfy a condition that we call \emph{supra-uniformity}.

\begin{Definition}[supra-uniformity]
\label{def:supre-uniform}
The distribution of ${p_i}$ is supra-uniform if for all $\lambda ,\gamma  \in [0,1]$ with $ \gamma \leq\lambda$ it holds $P({p_i} < \gamma\,|\,{p_i}\leq\lambda )\leq\gamma /\lambda$. We say that ${p_i}$ is supra-uniform if its distribution is supra-uniform.
\end{Definition}

Supra-uniformity is also known as the uniform conditional stochastic order (UCSO) \citep[defined by][]{Whitt1980, Whitt1982, KS1982,Ruschendorf1991} relative to the standard uniform distribution \SUD. It is well-known that this condition is implied if ${p_i}$ dominates \SUD in likelihood ratio order \citep[e.g.,][]{Whitt1980, Denuit2005}. \citealt{Whitt1980} shows that when the sample space is a subset of the real line and the probability measures have densities, then UCSO is equivalent to the monotone likelihood ratio (MLR) property (i.e. for every $y>x$ it holds $f(y)/g(y)\geq{}f(x)/g(x)$). In the case of \SUD (i.e. when $g$ is a constant) it is immediately clear that MLR is equivalent to the $p$-values having densities that are increasing on $(0,1)$, which is further equivalent to having cumulative distribution functions that are convex on $(0,1)$.

\begin{theorem} \label{thm:independent}

Let $\P$ be an MTP and $D$ be a decision rate (e.g., FWER of FDR) such that $D_{\P}\leq\alpha$ for $\alpha\in(0,1)$ whenever the $p$-values of the true hypotheses are independent and supra-uniformly distributed. If the $p$-values of the true hypotheses are independent and supra-uniform, then, for the conditionalized MTP $\Plam$ it holds that $D_{\Plam}\le\alpha$.
\end{theorem}

The proof of Theorem \ref{thm:independent} can be found in the Appendix. The basic idea behind the proof is to divide the space of $p$-values into orthants partitioned by the events ${[{p_i}\leq\lambda ]}$ versus ${[{p_i} > \lambda ]} $ for all $i$. Conditionally on each of these orthants, the FWER (or FDR) of $\Plam$ is at most $\alpha $. Therefore, the total FWER (or FDR) of $\Plam$ must also be at most $\alpha$. A similar argument is used by \citet{WD1986} in the context of order restricted inference.

Many popular multiple testing procedures satisfy the conditions of Theorem \ref{thm:independent}, since they only require a weaker condition $\mathrm{P}(p_i{}\leq{}c)\leq{}c$ in order to preserve type I error control. Consequently, for independent $p$-values the validity of the conditionalized versions of the methods by \cite{Holm1979}, \cite{Hommel1988}, \cite{Hochberg1988} for FWER control and by \cite{BH1995} for FDR control follows by Theorem \ref{thm:independent}.

\section{FWER control by the CBP under dependence: finitely many hypotheses}

Generalizing Theorem \ref{thm:independent} to the setting with dependent $p$-values is not trivial. In the next three sections of this paper we focus on the specific case of CBP, presenting several sufficient conditions the control of the FWER by the CBP. The overarching theme of these conditions is a requirement for $p$-values to be positively correlated. In light of the several special results obtained below we conjecture that, at least in the multivariate normal model, the CBP controls FWER whenever the $p$-values are positively correlated. Further justification for our conjecture is the extreme case when the $p$-values under the null are all identical at which point the proof of the control of the FWER by the CBP is trivial.

\subsection{Negative correlations: a counterexample}

To see that the requirement of independence of $p$-values in Theorem~\ref{thm:independent} cannot simply be dropped, consider a multiple testing problem with $m=2$ where $p_1$ and $p_2$ both have a \SUD and $p_1=1-p_2$. Assume $\lambda > 1/2$, since otherwise CBP is uniformly less powerful than the classical Bonferroni method. In this setting
\[
\FWERcb=\mathrm{P}(p_1\leq\lambda\alpha)+\mathrm{P}(p_2\leq\lambda\alpha)=2\lambda\alpha>\lambda.
\]
In other words, under the considered setting the CBP either fails to control FWER (with $\lambda>1/2$) or is strictly less powerful than the Bonferroni method (with $\lambda\leq1/2$).

\subsection{The bivariate normal case}

Proposition \ref{thm:bivnorm} below guarantees FWER control by the CBP for all $\alpha,\lambda\in(0,1)$ in the setting with two $p$-values corresponding to two bivariate zero-mean normally distributed test statistics with positive correlation. Denote as $\Phi$ the standard normal distribution function and as $I_2$ the $2\times 2$ identity matrix.




\begin{Proposition}
\label{thm:bivnorm}
Let $m=2$ and let $(X_1,X_2)'\sim{}N(0,\Sigma_\rho)$, where $\Sigma_\rho=\rho+(1-\rho)I_2$. Set $p_1=1-\Phi(X_1)$ and $p_2=1-\Phi(X_2)$. If $\rho\geq0$, then $\FWERcb\leq\alpha$.
\end{Proposition}

\subsection{Distributions satisfying the expectation criterion}

In the multivariate setting without independence of the $p$-values a number of conditions can be formulated which guarantee the control of the FWER by the CBP. One such sufficient condition is given in Lemma \ref{lem:expcrit}.

\begin{lemma}
\label{lem:expcrit}
 Let the $p$-values $p_1,\ldots,p_m$ have continuous distributions $F_1,\ldots,F_m$ that satisfy $(F_i(y) - F_i(x))\lambda \alpha \leq F_i(\lambda \alpha)(y-x)$ for any $x, y \in (0,\lambda \alpha)$ such that $x < y$, and let $P(R_m \le k\,|\,p_i=x)$ be increasing in $x$ for every $k=0,1,\ldots,m$ and $i=1,\ldots,m$ and let
\begin{align}
\label{eq:expcrit}
\sum_{i=1}^{m}{P(p_i\le\lambda\alpha)\,E(R_m^{-1}\,|\,p_i=\lambda\alpha)\le\alpha}.
\end{align}
Then $\FWERcb\le\alpha$.
\end{lemma}

We refer to the condition in (\ref{eq:expcrit}) as the \emph{expectation criterion}. Note that for positively associated $p$-values, small values of $p_i$ often occur together with small values of $R_m^{-1}$; hence for small $\lambda \alpha$, the summand in the expectation criterion tends to be small. The expectation criterion can be used to prove a general result on $p$-values arising from equicorrelated jointly normal test statistics formulated as Lemma \ref{lem:equic_norm}. Note that if $p_1\ldots,p_m$ are exchangeable and standard uniform,  (\ref{eq:expcrit}) simplifies to $E(R_m^{-1}\mid p_1=\lambda\alpha)\le (\lambda m)^{-1}$. If further $p_1\ldots,p_m$ are derived from jointly normally distributed test statistics, Lemma \ref{lem:equic_norm} gives a further simplification of the condition.

\begin{lemma}
\label{lem:equic_norm}
Assume that $(\Phi^{-1}(p_1),\ldots, \Phi^{-1}(p_m))'\sim{}N(0,\Sigma_\rho)$, where $\Sigma_\rho=\rho+\diag(1-\rho,\ldots,1-\rho)$ with $0\leq\rho<1$. Let
\begin{equation}
\label{eq:integralcrit}
\int_{-\infty }^{\infty }{\frac{\varphi(x)}{\Phi (\mu-\sqrt{\rho}x)}dx}\leq\lambda^{-1},
\end{equation}
where $\mu=\Phi^{-1}(\lambda)(1-\rho)^{-1/2}-\Phi^{-1}(\lambda\alpha)\rho(1-\rho)^{-1/2}$. Then $\FWERcb\le\alpha$.
\end{lemma}

A practical implication of Lemma \ref{lem:equic_norm} is that for a given setting $(\lambda, \alpha, \rho)$ the FWER control by the CBP can be verified numerically by evaluating the one-dimensional integral in (\ref{eq:integralcrit}). The results of our numerical analysis suggest that (\ref{eq:integralcrit}) holds for any $0\leq\rho<1$ and any $0<\lambda\leq1$ whenever $\alpha\leq0.368$. Moreover, we observed evidence that condition (\ref{eq:integralcrit}) is most certainly too strict: Certain combinations of $\alpha$, $\lambda$ and $\rho$ (e.g.\ $\alpha=0.7$, $\lambda=0.9$, $\rho=0.2$) exist for which (\ref{eq:integralcrit}) is violated, but simulations indicate that the CBP controls the FWER for all $\alpha$, $\lambda$ and $\rho \geq 0$ in the case of $p$-values arising from equicorrelated normals with nonnegative means and unit variances.

\subsection{Mixtures}


 Next we show that the control of the type I error rate by the CBP is preserved when distributions are mixed. The result below applies when the family of distributions being mixed is indexed by a one-dimensional parameter, however, it can be easily generalized to more complex families of distributions.

\begin{Proposition}
\label{lem:mixtures}
 Let $\mathcal{F}=\{F_w,w\in\R\}$ be a family of distribution functions. Assume that a decision rate $D$ and an MTP $\P$ satisfy $D_{\P} \leq\alpha$ whenever the joint distribution of the $p$-values is in $\mathcal{F}$. For any mixing density $g$, if the $p$-values $p_1,\ldots,p_m$ are distributed according to the mixture $F=\int_{-\infty}^\infty{}F_w g(w)dw$, then $D_{\P} \leq\alpha$.
\end{Proposition}

 Proposition \ref{lem:mixtures} is a general result that applies to many conditionalized procedures. For instance, in the case of the CBP, which controls the FWER whenever the $p$-values are independent and supra-uniform, the proposition guarantees the control of the FWER by the CBP also for mixtures of such distributions. Note that such mixtures may be correlated.

\section{FWER control by the CBP under dependence in large testing problems}

Finally, we give a sufficient conditions for FWER control by the CBP as the number of hypotheses $m$ approaches infinity. Suppose for a moment that the expectation of $R(\lambda)$ (i.e. the number of $p$-values below $\lambda$) is known. In such case one could use the alternative to CBP that rejects hypothesis $H_i$ whenever $p_i\leq\lambda\alpha/\E[R(\lambda)]$. If the $p$-values are supra-uniform then under arbitrary dependence this procedure controls FWER, since it holds
\begin{align*}
\mathrm{FWER}_{\mathrm{CBP'}}&\leq\sum\nolimits_{i\in T} P(p_i<\alpha\lambda/E[R(\lambda)]) \\
&\leq\sum\nolimits_{i\in T} P(p_i\leq\lambda)\,P(p_i<\alpha\lambda/E[R(\lambda)]\mid p_i\leq\lambda) \\
&\leq\sum\nolimits_{i\in T} P(p_i\leq\lambda)\,\alpha/E[R(\lambda)]
\leq\alpha.
\end{align*}
This suggests that the CBP should also control FWER for $m\to\infty$ whenever $R(\lambda)$ is an consistent estimator of $\E[R(\lambda)]$. This heuristic argument is formalized in Proposition \ref{thm:exp_bonf}, where $\plim$ denotes convergence in probability.

\begin{Proposition}
\label{thm:exp_bonf}
Let the $p$-values $p_1,\ldots,p_m$ have supra-uniform distributions and let $\plim_{m\to\infty}\,R/m=\eta$ and $\lim_{m\to\infty}\,E(R/m)=\eta$ for some $\eta\in\R$. Then $\limsup_{m\to\infty}\,\FWERcb\leq\alpha$.
\end{Proposition}

An application of Proposition~\ref{thm:exp_bonf} in a situation where correlations between $p$-values vanish with $m\to\infty$ leads to Corollary \ref{corol:asymptotic}.

\begin{Corollary}
\label{corol:asymptotic}
Denote $\rho_{ij}=\cor(\mbf{1}[p_i\leq\lambda],\mbf{1}[{{p}_{j}}\leq\lambda])$ and put $\rho_{ij}^+=\max\{0,\rho_{ij}\}$. Denote the average off-diagonal positive part of the correlations as
\begin{align*}
\bar\rho(m)=\frac{2}{m(m-1)}\sum\limits_{i=1}^{m-1}{\sum\limits_{j=i+1}^{m}{{{\rho }_{ij}^{+}}}}.
\end{align*}
If the $p$-values are supra-uniform and $\lim_{m\to\infty}\,E(R/m)=\eta$ for some $\eta\in\R$ and $\lim_{m\to\infty}\bar\rho(m)=0$, then $\limsup_{m\to\infty}\FWERcb\leq\alpha $.
\end{Corollary}

An example of the usage of Corollary~\ref{corol:asymptotic} for data analysis can be found in Section \ref{manifest}.

\section{FWER investigation - simulations}


Our conjecture is that the CBP controls the FWER in the case of positively correlated multivariate normal test statistics. To substantiate this, we conducted the following simulations. We generated the $p$-values as $p_i = \Phi(Z_i)$, with the 'test statistics'   $(Z_1, \ldots, Z_m)'\sim{}N(0,\Sigma) $ and $\Sigma = (\sigma_{ij})$ with each $\sigma_{ij} \ge 0$ and each $\sigma_{ii} =1$. The correlation matrices were generated in the following way. First, a covariance matrix was generated as $\Sigma = AA^T$, where the $a_{ij}$ were drawn randomly and independently from a standard normal distribution. If there was a negative covariance, then the smallest covariance was found, and the corresponding negative elements of $a_{ij}$ were set to 0. This was repeated until all covariances were nonnegative. Finally, the covariance matrix was scaled into a correlation matrix. For each $m\in\{1,2,3,4,5,6,7,8,9,10,15,20,25,50,75, 100\}$ we generated 100 correlation matrices $\Sigma$, and for each $\Sigma$ we conducted 10,000 simulations and computed the FWER for the CBP with $\alpha \in {0.05,0.10,\ldots,0.95}$ and $\lambda\in\{0.1,0.2,\ldots,0.9\}$. There were 6 combinations of $(\Sigma,\alpha, \lambda)$ with simulated FWER slightly above $\alpha$, but none of these differences were significant according to a binomial test with significance level 0.05. For $m \ge 6$ we found no cases with simulated FWER above $\alpha$.



In our simulations we also explored several multivariate settings with negative correlations (results not included). The simulations showed what was already suggested by the lack of control of the FWER for negative correlations in the bivariate case, namely that FWER is not necessarily controlled with negative correlations (especially for small $m$).

\section{Power investigation - simulations}

In this section we investigate the power performance of conditionalized tests relative to their non-conditionalized versions through simulations. We consider the following procedures: Bonferroni; {\v S}id{\'a}k \citep[attributed to Tippet by][p.~2433]{Davidov2011}; Fisher combination method based on the transformation $F=-2\sum_{i=1}^m\log {p_i}$  \citep[see][p.~2433]{Davidov2011}; the likelihood ratio (LR) procedure based on the theory of order restricted statistical inference of \citet{RobertsonWD}, using the chi-bar distribution with binomial weights; the ${I_+}$ statistic, based on the empirical distribution function \citep[p.~2433]{Davidov2011}; the Bonferroni plug-in procedure as defined by \citet{FG2009} based on the work of \citet{Storey2002} referred to as the FGS procedure.

It should be noted that neither Fisher's method nor the LR and ${I_+}$ procedures provide a decision for each individual hypothesis ${H_i}$. Instead they only allow a conclusion about the intersection hypothesis that all ${H_i}$ are true (i.e. the global null hypothesis). For this reason their usage is limited. Furthermore, both Fisher's and {\v S}id{\'a}k's method as well as the LR and ${I_+}$ procedures assume independence of the analyzed $p$-values. However, this assumption is often violated in practice, which limits the usage of these methods.

\subsection{Power as the number of true hypotheses increases}

For the power investigation the $p$-values were generated based on $m$ parallel \textit{z}-tests of null hypotheses of type $H_0:\mu_i\geq0$, each based on a sample of size $n$. The $p$-values were calculated as $p_i=\Phi(X_i)$, with $var(X_i) = 1$ and noncentrality parameter $E(X_i)=\mu_i \sqrt{n}$. To each set of $p$-values we applied the conditionalized and the ordinary versions of the considered testing procedures at the overall significance level of $\alpha=.05$. Conditionalizing was applied with $\lambda=0.5$. A number of combinations in terms of noncentrality parameter, hypothesis count and proportion of false hypotheses was considered. For each combination we performed 10,000 replications.

Fig.~\ref{fig:figureM} shows the results of a simulation where the number of false hypotheses is fixed while the number of true hypotheses increases. For the false hypotheses the value of the noncentrality parameter was set at $-2$, while for the true null hypotheses it was set at $2$. This illustrates that the power decreases rapidly with the number of true hypotheses for most non-conditionalized procedures. The only exception to this is the LR procedure. In contrast, for all of the considered conditionalized procedures the power decreases much more slowly. This shows that, with the exception of the LR procedure, the conditionalization substantially improves the power performance of the considered procedures in this setting. Among the procedures that permit a per-hypothesis decision (i.e. Bonferroni, {\v S}id{\'a}k, and FGS) it is the conditionalized FGS procedure that shows the highest power.

Fig.~\ref{fig:figureP2} illustrates the influence of conditionalization on the performance of the Bonferroni and FGS methods in a setting where the percentage of true hypotheses increases while the total number of hypotheses remains fixed. The figure shows that the conditionalized FGS procedure is the overall best performing procedure among the four.

\subsection{Power in pairwise comparisons of ordered means}
\label{ordered}

Consider a series of independent sample means $y_i \sim N(\mu_i,\sigma^2/n)$ with the compound hypothesis $H:\mu_1 \leq\mu_2\leq\ldots\leq\mu_k$. An analysis method specifically designed for this setting is the isotonic regression \citep{RobertsonWD}, although this method does not allow to deduce specifically which pairs $(\mu_i,\mu_j)$ violate the ordering specified by the null hypothesis. Alternatively, the $k(k-1)/2$ individual hypotheses $H_{ij}:\mu_i\leq\mu_j$ with $i < j$ can be analyzed using one-sided \textit{t}-tests, and the conditionalized Bonferroni or the conditionalized FGS procedures can be applied. The average correlation between the $p$-values vanishes as $m \to \infty$, thus the asymptotic control of the FWER by the CBP follows by Corollary~\ref{corol:asymptotic}. The simulations below indicate that the FWER is in fact controlled even for the small hypothesis counts. 

The means in this simulation were modeled as $\mu_{i+1}=\mu_i+\delta$ for $i=1,2,\ldots,k-2$, and $\mu_k = \mu_1$. Thus, most means satisfy the ordering of the hypothesis, but the last mean violates it. We used $\sigma = 1$, $n = 10$ and set $\lambda = 0.5$. Fig.~\ref{fig:figurePairs} shows the results for $k = 20$ and $k = 5$ respectively. At $\delta = 0$, it is observed that all four procedures exhibit $\FWER$ below $\alpha$. For both $k = 5$ and $k = 20$, the two conditionalized procedures perform essentially as good or better than their non-conditionalized counterparts across the whole range of $\delta \in [0.1, 3]$. Note that the same results would be obtained with, for example, $n = 90$ and $\delta \in [1/30, 1]$.

\section {Examples with real data}
\subsection{Example 1 (detecting adverse effects in meta-analysis)}

Suppose the effect of a medical or psychological treatment is investigated in a meta-analysis of $m$ studies in distinct populations, and that a one-sided \textit{t}-test of $H_i^*: \vartheta_i \le 0$ is conducted in each population, yielding $p$-values $p^*_i$, where $\vartheta_i > 0$ indicates a positive, beneficial effect of the treatment in population $i$. In a meta-analysis, one would usually test a weighted average of the effects, say $\bar{\vartheta}=\sum\nolimits_iw_i\vartheta_i$. However, even if the average effect $\bar{\vartheta}$ is positive, there can be populations in which the local effect $\vartheta_i$ is negative. It would be wise to test for such adverse effects, as the treatment should not be recommended in such populations. This means that one also has to test the opposite hypothesis, $H_i: \vartheta_i \ge 0$, in each population, yielding $p$-values $p_i = 1 - p^*_i$ . This is a problem of the form considered in this article.

Under the classical \textit{t}-test applied in the context of interval hypothesis testing, each \textit{t}-statistic has a non-central \textit{t}-distribution with the non-centrality parameter determined by the true value of the expectation $\vartheta_i$. Under the principle of least favorability the $p$-values are obtained using the central \textit{t}-distribution. It is well-known that the ratio of densities of \textit{t}-distributions with different non-centrality parameters (and equal degrees of freedom) is monotone \citep{Kruskall1954,Lehmann1955}, and from this it follows that they are ordered in likelihood ratio. This implies that the supra-uniformity of Definition \ref{def:supre-uniform} applies \citep{Whitt1980, Whitt1982, Denuit2005}. Consequently, assuming that the $p$-values are independent, by Theorem \ref{thm:independent} it follows that in this setting the CBP and the conditionalized FGS procedure control the FWER, while the conditionalized Benjamini-Hochberg procedure \citep{BH1995,BY2001} controls the FDR.

\subsection{Example 2. Detecting substandard organizations in quality benchmarking}

Several countries have developed programs in which the quality of public organizations such as schools or hospitals is assessed. As stated by \citet{Ellis2013}, "such research can consist of large-scale studies where dozens [3], hundreds [4], or thousands [5, 6] of organizations are compared on one or more measures of performance or quality of care, on the basis of a sample of clients or patients from each organization". A goal of such programs is to identify under-performing organizations. For example, in the Consumer Quality Index (CQI) program of the Netherlands, the questionnaire used in 2010 to evaluate the short-term ambulatory mental health and addiction care organizations contained a question whether it was a problem to contact the therapist by phone in the evening or during the weekend in case of emergency. Now suppose that a minimum standard of 90\% satisfaction rate is imposed. Under such standard in each organization 90\% or more of the patients should answer that contacting the therapist outside office hours was not a problem. Investigating whether hospitals satisfy this minimum standard can be done using the binomial test within each hospital with the null hypothesis of type $H_0:\pi\ge.90$ where $\pi$ denotes the success rate. The advisory statistics team debated the question of whether a correction for multiplicity for all hospitals is required in such analyses. The arguments against correcting for all hospitals were motivated by the expected loss of power associated with multiplicity correction for all hospitals in non-conditionalized MTPs. The advantage of using a conditionalized MTP in such setting is that the presence of organizations that score high above the minimum standard does not exacerbate the severity of the multiple testing problem and much of the power is preserved even with many high-performing hospitals included in the analysis.

\subsection{Example 3. Testing for manifest monotonicity in IRT}\label{manifest}

In Mokken scale analysis it is recommended to test manifest monotonicity \citep{Ark2007}. With $k$ items to be tested suppose that the variables $X_1,\ldots,X_k$ indicate correctness of response for the $k$ items (with $X_i=1/0$ indicating a correct/incorrect answer for the $i$-th item). Denote the \textit{rest score} of the $i$-th item as $X_{-i}=(\sum_{j=1}^{k}X_j)-X_i$. A question of interest is whether $\pi_{ij}:=P(X_i=1\,|\,X_{-i}=j)$ is a nondecreasing function of $j$ within each item $i$. This leads to testing the $k(k-1)/2$ pairwise hypotheses $\pi_{ij'}\leq\pi_{ij}$ for $j'<j$ \citep{Ark2007}. In the subtest E of the Raven Progressive Matrices test in the data set reported by \citet{VE2000} we obtained the following result. For item 11, there were 21 pairs of rest score groups that had to be compared - small adjacent groups were joined together by the program. There were 4 violations with a maximum  \textit{z}-statistic of 2.33, yielding an unmodified $p$-value $p = 0.010$. If no multiplicity correction is performed the probability of false rejection for each item undesirably increases with the number of rest score groups. The classical Bonferroni correction yields the adjusted $p$-value of $p'=0.010\times 21=0.21$, while the CBP yields the adjusted $p$-value of $p''=0.010\times 4=0.04$. As the number of items increases, the number of pairwise comparisons increases, but the average correlation between the $z$-statistics vanishes. In this situation, Corollary~\ref{corol:asymptotic} implies that the FWER is asymptotically under control, while the simulations of Section \ref{manifest} indicate that FWER control is already achieved with small hypothesis count. Thus, both the classical Bonferroni correction and the CBP control the FWER, but the CBP yields the smallest $p$-value.

\subsection{Example 4. Testing for nonnegative covariances in IRT}

In Mokken scale analysis and, more generally in monotone latent variable models, it is required that the test items have nonnegative covariances with each other \citep{Mokken1971,HR1986}. Two approaches are possible in item selection with this requirement. One approach is to retain only items with significantly positive covariances, and the other approach is to delete items with significantly negative covariances. We consider the latter approach here. The distribution of the standardized sample covariances converges to a normal distribution with increasing sample size, which suggests that the CBP might control the FWER in this setting. We investigated this further, both analytically and with simulations. Both approaches suggest that the FWER is indeed under control, but we intend to report the details of this in a psychometric journal. Here, consider only briefly an example. We have deployed this procedure on an exam with 78 multiple choice questions. There were 3003 covariances between items, of which 280 were negative. The smallest unadjusted $p$-value was 0.000243. The Bonferroni corrected $p$-value is 0.73, while the CBP with $\lambda=0.5$ yields a $p$-value of 0.14.

\section{Discussion}

We have proposed a very simple and general method, called conditionalization, to deal with  the presence of inflated $p$-values in multiple testing problems. Such $p$-values often arise in practice for instance when interval hypotheses are tested. We suggest to discard all hypotheses with $p$-values above a pre-chosen constant $\lambda$ (typically 0.5 or higher), and to divide the remaining $p$-values by $\lambda$ before applying the multiple testing procedure of choice. For independent $p$-values, we have proven that the conditionalized procedure controls the same error rate as the original procedure, provided null $p$-values are supra-uniform (i.e.\ dominate the standard uniform distribution in likelihood ratio order). As a rule of thumb, conditionalized procedures can be expected to be more powerful than their ordinary, non-conditionalized counterparts if there are more true hypotheses with inflated $p$-values (i.e.\ with true parameter values deep inside the null hypothesis) than there are false null hypotheses. The power gain achieved by conditionalizing can be substantial, especially for adaptive procedures that incorporate an estimate of the proportion of true null hypotheses.
 
For the case of the conditionalized Bonferroni procedure (CBP) we conjecture that the CBP is valid when the $p$-values are positively correlated. For this case we have given several sufficient conditions for FWER control by the CBP. We accompanied these results with an extensive simulation study and the results give supporting evidence for our conjecture. Nonetheless, a proof of our conjecture still eludes us and thus remains for future research. We have shown that it is not universally valid for negatively correlated variables, however. Other topics that in our opinion deserve further attention are the question of how to optimally choose the value for the cut-off parameter (i.e. $\lambda$) and whether the procedure is valid when the $p$-values are based on discretely distributed test statistics, since these typically do not fulfill the supra-uniformity condition of Theorem \ref{thm:independent}.

We believe that this paper makes a strong case for the usage of the conditionalized multiple testing procedures since they mitigate the loss of power typically associated with multiple testing procedures on inflated $p$-values and thus make it more attractive for researchers to formulate their scientific questions in terms of interval hypotheses. In light of the fact that shifting the focus towards interval hypotheses has been advocated as one of the solutions to get out of the current "$p$-value controversy" \citep{Wellek2017} this likely makes conditionalization a very powerful method of analysis.

\appendix
\section{Proofs}

The following notation is used. For $\mbf{p}=({p_1},\ldots,{p_m})$ and a subset $K \subseteq \{ 1,\ldots,m\} $, denote by $\mbf{p}_{\!K}$ be the subvector of components ${p_i}$ with $i\in{}K$. The index set of true hypotheses is denoted by $\mathcal{T}$, that is: $i \in \mathcal{T} \Leftrightarrow ({H_i}$ is true). For a scalar $\lambda $ we write $\mbf{p}_{\!K}\leq\lambda $ or $\mbf{p}_{\!K}>\lambda $ if these inequalities hold for component-wise. ${\mbf{U}}=({U_1},\ldots,{U_m})$ denotes a random vector with independent components that follow \SUD.

\subsection{Independence case}
\begin{proof}[Proof of Theorem 3.2]
For any set $K \subseteq \mathcal{T}$, let $\bar K = \mathcal{T} - K$, and define the orthant ${G_K} = [{{\mbf{p}}_K}\leq\lambda,{{\mbf{p}}_{\bar K}} > \lambda ]$. Conditionally on each ${G_K}$, the modified $p$-values ${p'_i} = {p_i}/\lambda$, $i\in\{1,\ldots,m\}$ are independent, and the distribution of each ${p'_i}$, $i\in\mathcal{T}$ stochastically dominates \SUD. Because $\P$ controls the decision rate under these circumstances, we may conclude that $E(D_{\Plam}|G_K) \leq \alpha$. Consequently, by the law of total expectation, $E(D_{\Plam})\leq\alpha$.
\end{proof}

\subsection{Bivariate normal case}

In this section we formulate three additional lemmaas which together immediately imply the validity of Proposition 1.

\begin{lemma}
\label{lem.bivariate_PQD}
Let $m=2$ and let $p_1,p_2$ be marginally standard uniformly distributed under the null hypothesis. Put $X_1=\Phi^{-1}(1-p_1)$ and $X_2=\Phi^{-1}(1-p_2)$ and let $X_1,X_2$ be positive quadrant dependent. Let $\alpha,\lambda\in(0,1)$ be such that $1-(1-\tfrac{\lambda\alpha}{2})^2+2(1-\lambda)\lambda\alpha\;\leq\;\alpha$. Then $\FWERcb\leq\alpha$.
\end{lemma}

\begin{proof}
For fixed $\alpha\in(0,1)$ and $\lambda\in(0,1)$ the CBP rejects $H_i$ whenever $p_i\leq\lambda\alpha/R_\lambda$, where $R_\lambda=I\{p_1\leq\lambda\}+I\{p_2\leq\lambda\}$. The corresponding FWER can be written as $\FWERcb=1-P(A^p)+2[P(B^p)-P(C^p)]$, where $A^p=\{p_i\in(\tfrac{\lambda\alpha}{2},1),i=1,2\}$, $B^p=\{p_1\in(\lambda,1),p_2\in(0,\lambda\alpha)\}$ and $C^p=\{p_1\in(\lambda,1),p_2\in(0,\tfrac{\lambda\alpha}{2})\}$. Since $P(C^p)\geq0$, it holds $\FWERcb\leq1-P(A^p)+2P(B^p)$. Since $A^p$ is an "on-diagonal" quadrant, with positive quadrant dependence the probability $P(A^p)$ is \emph{minimized} under independence, when its probability is $P^\bot(A^p)=(1-\tfrac{1}{2}\lambda\alpha)^2$. Analogously, $B^p$ is an "anti-diagonal" quadrant, which means that $P(B^p)$ is \emph{maximized} under independence, thus $P(B^p)\leq{}P^\bot(B^p)=(1-\lambda)\lambda\alpha$. Consequently, $\FWERcb\leq1-(1-\tfrac{\lambda\alpha}{2})^2+2(1-\lambda)\lambda\alpha$.
\end{proof}

Solving this inequality with respect to $\alpha$ and $\lambda$ yields a set of combinations of $\alpha$ and $\lambda$ for which the CBP controls FWER under positive dependence. The permissible ranges are depicted in Figure \ref{fig.lambda_ok_rough_QPD}. Note that if $\lambda\leq\tfrac{1}{2}$ it holds trivially $\FWERcb\leq\alpha$ since in such case $\FWERcb$ is dominated by the FWER of the classical Bonferroni method. The lemma requires that the test statistics are positive quadrant dependent, which under the bivariate normal model with correlation $\rho$ is equivalent to $\rho\geq0$.

\begin{lemma}
\label{lem.bivariate_normal}
Let $n=2$ and let $p_1,p_2$ be marginally standard uniformly distributed under the null hypothesis. Put $X_1=\Phi^{-1}(1-p_1)$ and $X_2=\Phi^{-1}(1-p_2)$ and let $(X_1,X_2)'\sim{}N(0,\Sigma_\rho)$ under the null hypothesis, where $\Sigma_\rho=\rho+\diag(1-\rho,1-\rho)$.
Let $\alpha,\lambda\in(0,1)$ be such that $\alpha\lambda\leq\frac{2}{3}$. If $\rho\geq0$, then $\FWERcb\leq\alpha$.
\end{lemma}

\begin{proof}
Analogously to the proof of Lemma \ref{lem.bivariate_PQD}, we can write $\FWERcb$ as
\begin{align}
\label{eqn.Fs}
\FWERcb&\;=\;1-P_{\!\!\rho}(A^p)+2[P_{\!\!\rho}(B^p)-P_{\!\!\rho}(C^p)],
\end{align}
where $A^p=\{p_i\in(\tfrac{\lambda\alpha}{2},1),i=1,2\}$, $B^p=\{p_1\in(\lambda,1),p_2\in(0,\lambda\alpha)\}$, $C^p=\{p_1\in(\lambda,1),p_2\in(0,\tfrac{\lambda\alpha}{2})\}$, where we added the lower index $\rho$ into the notation $P_{\!\!\rho}$ to signify that the probability is a function of $\rho$. We proceed to show that $\FWERcb$ is a decreasing function in $\rho\in[0,1)$ whenever $\alpha\lambda\leq\frac{2}{3}$, which we do by differentiating $P_{\!\!\rho}(A^p)$, $P_{\!\!\rho}(B^p)$ and $P_{\!\!\rho}(C^p)$ with respect to $\rho$.

By \citet{Tong90} (page 191) the derivative of the bivariate normal distribution function $F_\rho(x)$ with respect to $\rho$ equals its density at $x$. Consequently, $\tfrac{\partial}{\partial\rho}\,P_{\!\!\rho}(A^p)=f_\rho(z_2,z_2)$, $\tfrac{\partial}{\partial\rho}\,P_{\!\!\rho}(B^p)=-f_\rho(z_0,z_1)$, $\tfrac{\partial}{\partial\rho}\,P_{\!\!\rho}(C^p)=-f_\rho(z_0,z_2)$, where $f_\rho$ is the density function of the unit-variance bivariate normal distribution
\begin{align*}
f_\rho(x,y)&\;=\;\tfrac{1}{2\pi}(1-\rho^2)^{-1/2}\,\exp(-\tfrac{(x-\mu_1)^2-2\rho{}(x-\mu_1)(y-\mu_2)+(y-\mu_2)^2}{2(1-\rho^2)}),
\end{align*}
and $z_0=\Phi^{-1}(1-\lambda)$, $z_1=\Phi^{-1}(1-\lambda\alpha)$, $z_2=\Phi^{-1}(1-\tfrac{\lambda\alpha}{2})$. Therefore, $\tfrac{\partial}{\partial\rho}\,\FWERcb=-f_\rho(z_2,z_2)-2[f_\rho(z_0,z_1)-f_\rho(z_0,z_2)]$, which in turn means that $\tfrac{\partial}{\partial\rho}\,\FWERcb\leq0$ whenever
\begin{align}
\label{eqn.Fs_via_devatives_condition}
\frac{f_\rho(z_2,z_2)}{f_\rho(z_0,z_2)}+2\frac{f_\rho(z_0,z_1)}{f_\rho(z_0,z_2)}\;\geq\;2.
\end{align}
The conditional distribution of $X_2\,|\,X_1=z_0$ is $N(\mu_1+\rho(z_0-\mu_2), 1-\rho^2)$, which has density
$g_\rho(x;z_0)=(2\pi(1-\rho^2))^{-1/2}\exp(-\tfrac{1}{2}\,(x-\rho{}z_0-(\mu_1-\rho\mu_2))^2(1-\rho^2)^{-1})$. Consequently,
\begin{align*}
\frac{f_\rho(z_0,z_1)}{f_\rho(z_0,z_2)}&\;=\;\frac{g_\rho(z_1;z_0)}{g_\rho(z_2;z_0)}\;=\;\exp(-(z_1-\rho{}z_0-\mu_1+\rho\mu_2)^2+(z_2-\rho{}z_0-\mu_1+\rho\mu_2)^2).
\end{align*}
Similarly, the conditional distribution of $X_1\,|\,X_2=z_2$ is $N(\mu_2+\rho(z_2-\mu_1), 1-\rho^2)$, and therefore
\begin{align*}
\frac{f_\rho(z_2,z_2)}{f_\rho(z_0,z_2)}&\;=\;\frac{g_\rho(z_2;z_2)}{g_\rho(z_0;z_2)}\;=\;\exp(-(z_2-\rho{}z_2-\mu_2+\rho\mu_1)^2+(z_0-\rho{}z_2-\mu_2+\rho\mu_1)^2).
\end{align*}
Under the null hypothesis we have $\mu_1=\mu_2=0$. Define
\begin{align}
\label{eqn.h_rho}
h(\rho)&\;=\;\exp(-(z_2-\rho{}z_2)^2+(z_0-\rho{}z_2)^2)+2\exp(-(z_1-\rho{}z_0)^2+(z_2-\rho{}z_0)^2).
\end{align}
Then (\ref{eqn.Fs_via_devatives_condition}) is equivalent to $h(\rho)\geq2$. Differentiating $h(\rho)$ with respect to $\rho$ yields
\begin{align*}
h'(\rho)
&\;=\;2z_2(z_2-z_0)\exp(-(z_2-\rho{}z_2)^2+(z_0-\rho{}z_2)^2)-4z_0(z_2-z_1)\exp(-(z_1-\rho{}z_0)^2+(z_2-\rho{}z_0)^2).
\end{align*}
Since $z_2\geq0$, $z_2\geq{}z_1$ and $z_0\leq0$, it holds $h'(\rho)\geq0$. In other words, $h(\rho)$ is minimized at $\rho=0$, where $h(0)=\exp(z_0^2-z_2^2)+2\exp(z_2^2-z_1^2)$.
Clearly, for any $|z_1|\leq|z_2|$ it holds $h(\rho)\geq2$ and the inequality (\ref{eqn.Fs_via_devatives_condition}) is satisfied. Since $\Phi^{-1}$ is strictly monotone and symmetric about $\tfrac{1}{2}$, finding the largest $\alpha$ for a given $\lambda$ such that $|z_1|\leq|z_2|$ leads to the inequality $\tfrac{1}{2}-(1-\lambda\alpha)\leq1-\tfrac{\lambda\alpha}{2}-\tfrac{1}{2}$, which is equivalent to $\lambda\alpha\leq\tfrac{2}{3}$.
\end{proof}

As it turns out, Lemmas \ref{lem.bivariate_PQD} and \ref{lem.bivariate_normal} together cover almost all combinations of $\alpha,\lambda\in(0,1)$. In Figure \ref{lem.bivariate_normal} the right plot shows the area (white) which is not covered by the two lemmas. Next we close the "gap" left uncovered by the two lemmas.

\begin{lemma}
\label{lem.bivariate_normal_gap}
Let $n=2$ and let $p_1,p_2$ be marginally standard uniformly distributed under the null hypothesis. Put $X_1=\Phi^{-1}(1-p_1)$ and $X_2=\Phi^{-1}(1-p_2)$ and let $(X_1,X_2)'\sim{}N(0,\Sigma_\rho)$ under the null hypothesis, where $\Sigma_\rho=\rho+\diag(1-\rho,1-\rho)$.
Let $\alpha,\lambda\in(0,1)$ be such that $\alpha\lambda\geq\frac{2}{3}$ and $1-(1-\tfrac{\lambda\alpha}{2})^2+2(1-\lambda)\lambda\alpha\;\geq\;\alpha$. If $\rho\geq0$, then $\FWERcb\leq\alpha$.
\end{lemma}

\begin{proof}
It can be easily verified that $(\lambda,\alpha)=(\tfrac{2}{3},1)$ and $(\lambda,\alpha)=(\tfrac{3}{4},\tfrac{8}{9})$ are the two points for which the two inequalities in the lemma simultaneously turn into equalities. Moreover, for any $\lambda,\alpha\in(0,1)$ such that $\lambda\alpha\geq0.69$ it holds $1-(1-\tfrac{\lambda\alpha}{2})^2+2(1-\lambda)\lambda\alpha\;\leq\;\alpha$. Consequently, we only need to show that $\FWERcb\leq\alpha$ for $\lambda,\alpha\in\Omega$, where $\Omega=\{(\alpha,\lambda)\in\R^2:\alpha\in[\tfrac{8}{9},1),\lambda\in[\tfrac{2}{3},\tfrac{3}{4}],\tfrac{2}{3}\leq\lambda\alpha\leq0.69\}$. As discussed in the proof of Lemma \ref{lem.bivariate_normal}, it is sufficient to show that $h(\rho)\geq2$ with $h(\rho)$ defined in (\ref{eqn.h_rho}). It can be easily shown that on $\Omega$ we have $0.398<z_2<0.431$ and $-0.496<z_1<-0.43$ and both $z_1$ and $z_2$ are decreasing in $\lambda\alpha$ and so is $z_2^2-z_1^2$ with the minimum above $-0.087$. Consequently, $-2\exp(z_2^2-z_1^2)\geq1.83$ and since it also holds also $\exp(z_0^2-z_2^2)\geq\exp(z_0^2-z_2^2)\geq\exp(z_0^2-z_2^2)>0.83$, we get $h(0)\geq2$. Since it was already shown in the proof of Lemma \ref{lem.bivariate_normal} that $h(\rho)$ is decreasing in $\rho$, this concludes the proof.
\end{proof}

\begin{proof}[Proof of Proposition 1]
The proposition is an immediate corollary of Lemmas \ref{lem.bivariate_PQD}--\ref{lem.bivariate_normal_gap}.
\end{proof}

\subsection{Expectation criterion}
\begin{proof}[Proof of Lemma 4.1: expectation criterion]
It is sufficient to consider the cases where all tested hypotheses are true, since adding false hypotheses to the test cannot increase the FWER. Divide the interval $(0, \lambda \alpha]$ into intervals $B_k = (b_{k+1}, b_k]$ with $b_k = \lambda \alpha / k$ for $k=1,\ldots,m$, and $b_{m+1}=0$. For each $H_i$, denote with $E_i$ the probability to reject $H_i$. It is given by
\begin{align*}
E_i&\;=\;P(p_i \le \lambda \alpha / (R \lor 1))
\;=\;P(R \le \tfrac{\lambda \alpha}{p_i}, p_i \le \lambda \alpha))
\;=\;\sum_{k=1}^{m}{P(R \le \tfrac{\lambda \alpha}{p_i}\,|\,p_i \in B_k)P(p_i \in B_k)}.
\end{align*}
Since $P(R\le{}k\,|\,p_i=x)$ is assumed to be increasing in $x$, we get
\begin{align*}
E_i&\;\leq\;\sum_{k=1}^{m}{P(R \le k\,|\,p_i =\lambda \alpha)P(p_i \in B_k)}\\
&\;\leq\;\sum_{k=1}^{m}{P(R \le k\,|\,p_i =\lambda \alpha)(b_k-b_{k+1})\frac{F_i(\lambda \alpha)}{\lambda \alpha}}\\
&\;=\;\sum_{k=1}^{m}{P(R = k\,|\,p_i =\lambda \alpha)\frac{1}{k} F_i(\lambda \alpha)}\\
&\;=\;E(R^{-1}\,|\,p_i =\lambda \alpha)P(p_i \le \lambda \alpha).
\end{align*}
Therefore, $\FWERcb\le\sum_{i=1}^{m}{E(R^{-1}\,|\,p_i=\lambda\alpha)P(p_i\le\lambda\alpha)}\leq\alpha$ by the assumptions of the lemma.
\end{proof}

\subsection{Equicorrelated normal case}

For the proof of Lemma 4.2 we need Lemma \ref{lem:binomial}.
\begin{lemma}
\label{lem:binomial}
If $X$ is a random variable with binomial $(n, p)$ distribution, then
\begin{align*}
E(X+1)^{-1}\;=\;\frac{(1-(1-p)^{n+1})}{(n+1)p} \le \frac{1}{(n+1)p}.
\end{align*}
\end{lemma}

\begin{proof}
Using $\frac{1}{k+1}\binom{n}{k} = \frac{1}{n+1}\binom{n+1}{k+1}$
we obtain
\begin{align*}
E(X+1)^{-1}&\;=\;\frac {1}{n+1} \sum_{k=0}^{n}\tbinom{n+1}{k+1}p^k (1-p)^{n-k}
\;=\;\frac{1}{(n+1)p} \sum_{k=1}^{n+1}\tbinom{n+1}{k}p^k (1-p)^{n+1-k},
\end{align*}
where the last sum corresponds to the binomial distribution with parameters $n+1$ and $p$, and is thus upper-bounded by 1.
\end{proof}

\begin{proof}[Proof of Lemma 4.2: equicorrelated normal case]
Note that the $p$-values are standard uniform, thus their distribution functions $F_i$ satisfy the condition $(F_i(y) - F_i(x))\lambda \alpha \leq F_i(\lambda \alpha)(y-x)$ of Lemma 4.1. Moreover, $P(R_m \leq k\,|\,p_i)$ is increasing in $p_i$ by Theorem 4.1 of \citet{KR1980}, since the $p_i$ are multivariate totally positive of order 2 and the function 
$\phi(p_1, \ldots, p_m):=\boldsymbol{1}[{(\sum_{i=1}^m\boldsymbol{1}\{p_i\le\lambda\}) \le k]}$ is increasing. It remains to prove that the the expectation criterion is satisfied. We may write  ${{Z}_{i}}=\sqrt{\rho }\,\Theta +\sqrt{1-\rho }\,{{\varepsilon }_{i}}$, where $\Theta,{{\varepsilon }_{1}},\ldots,{{\varepsilon }_{n}}$ are independent standard normal variables. Define $z=\Phi^{-1}(\lambda\alpha)$ and $\Theta_z=(\Theta - \sqrt{\rho} z)/\sqrt{1-\rho}$. Immediately, $\Theta_z$ has a standard normal distribution conditionally on the event $[Z_i=z]$, and hence on $[p_i=\lambda \alpha]$. Define $X_i = \boldsymbol{1}_{\{p_i \leq \lambda\}}$. Conditionally on $\Theta_z$, the $X_i$ are independent Bernoulli variables with success probability $P(X_i=1\,|\,\Theta_z)=\Phi(\mu-\sqrt{\rho}\Theta_z)$. Therefore the conditional distribution $R\,|\,[\Theta_z,p_i=\lambda \alpha]$ is equal to the conditional distribution $R\,|\,[\Theta_z,X_i=1]$, and Lemma \ref{lem:binomial} yields $E(R^{-1}\,|\,\Theta_z,p_i=\lambda\alpha)\le(m\Phi(\mu-\sqrt{\rho}\Theta_z))^{-1}$. Taking the expectation over $\Theta_z$ and using (2) yields $E(R^{-1}\,|\,p_i=\lambda\alpha)\le(m\lambda)^{-1}$, which in turn implies that the expectation criterion is satisfied. The conclusion then follows by Lemma 4.1.
\end{proof}

\subsection{Mixtures}
\begin{proof}[Proof of Proposition 2: mixtures]
By the law of total expectation it follows that $E(D_{\P}) = E(E(D_{\P}|w)) \leq \alpha$.
\end{proof}

\subsection{Asymptotic control}
\begin{proof}[Proof of Proposition 3: asymptotic case]
For a given $\varepsilon '>0$ we prove that, for $m$ sufficiently large, $\FWER\le\alpha+\varepsilon '$. First, note that there is an $\varepsilon >0$ such that
$\frac{\eta +\varepsilon}{\eta -\varepsilon }\le 1+\tfrac{1}{2}\varepsilon'$. Moreover, for $m$ large enough we have, as a consequence of the convergence in probability, it holds $P(|R/m-\eta|\ge\varepsilon)\le\tfrac{1}{2}\varepsilon'$ and $E(R/m)\le\eta+\varepsilon$. Then
\begin{align*}
\FWER&\;=\;P\Big(\textstyle\bigcup\limits_{i=1}^{m}{p_i<\frac{\alpha \lambda }{R}}\Big) \\
&\;=\;P\Big(\textstyle\bigcup\limits_{i=1}^{m}{p_i<\frac{\alpha\lambda}{R}},\frac{R}{m}\ge\eta-\varepsilon\Big)
+P\Big(\textstyle\bigcup\limits_{i=1}^{m}{p_i<\frac{\alpha\lambda}{R}},\frac{R}{m}<\eta-\varepsilon\Big)\\
&\;\le\;P\Big(\textstyle\bigcup\limits_{i=1}^{m}{p_i<\frac{\alpha\lambda}{(\eta-\varepsilon )m}}\Big)+P\big(\frac{R}{m}<\eta -\varepsilon\big) \\
&\;\le\; \sum\limits_{i=1}^{m}{P(p_i<\tfrac{\alpha \lambda }{(\eta -\varepsilon )m})}+\frac{\varepsilon '}{2} \\
&\;=\;\sum\limits_{i=1}^{m}{P(p_i<\tfrac{\alpha \lambda }{(\eta -\varepsilon )m}\,|\,p_i\le \lambda )}\,P(p_i\le \lambda )+\frac{\varepsilon '}{2} \\
&\;\le\; \sum\limits_{i=1}^{m}{\frac{\alpha }{(\eta -\varepsilon )m}}\,P(p_i\le \lambda )+\frac{\varepsilon '}{2} \\
&\;=\; \frac{\alpha }{\eta -\varepsilon }E(R/m)+\frac{\varepsilon '}{2} \\
&\;\le\; \alpha\frac{\eta +\varepsilon }{\eta -\varepsilon }+\frac{\varepsilon '}{2} \\
&\;\le\; \alpha+\varepsilon'.
\end{align*}
\end{proof}

\begin{proof}[Proof of Corollary 5.1]
From
\begin{align*}
\var(R)&\;=\;\sum\limits_{i=1}^{m}{\var(\mathbf{1}[p_i\le\lambda ])
+\sum\limits_{i=1}^{m}{\sum\limits_{j=1,j\ne i}^{m}{\text{cov}(\mathbf{1}[p_i\le\lambda],\mathbf{1}[{{p}_{j}}\le \lambda ])}}}
\end{align*}
it follows $\var(R/m)\le \frac{1}{4}(\frac{1}{m}+\bar\rho_m)$, where the right-hand side goes to 0 as $m\to\infty$. The rest follows by Proposition 3.
\end{proof}

\bibliography{julesreferences}

\begin{figure}[!p]
\centering
\includegraphics[width=0.48\linewidth]{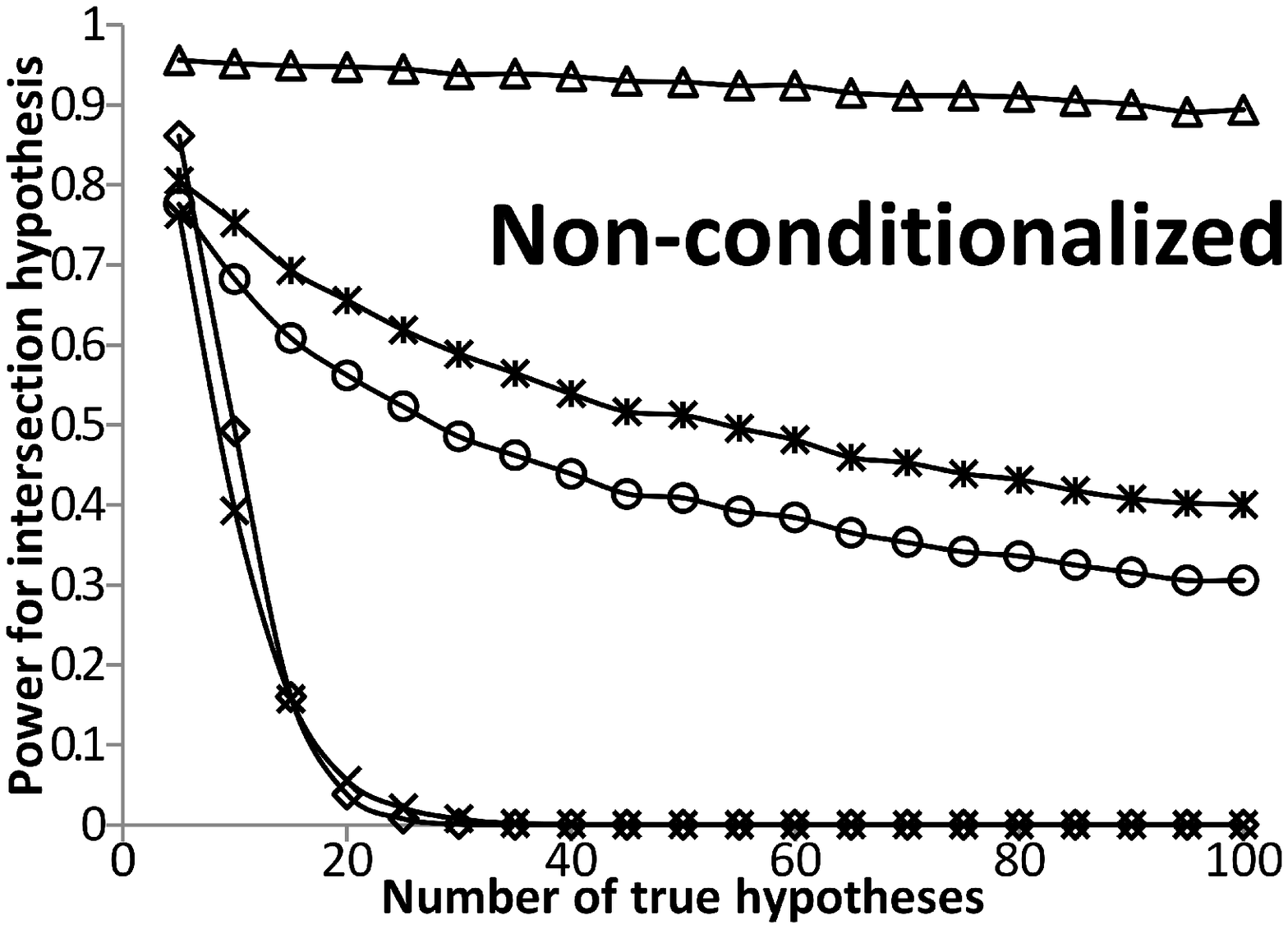}
\includegraphics[width=0.48\linewidth]{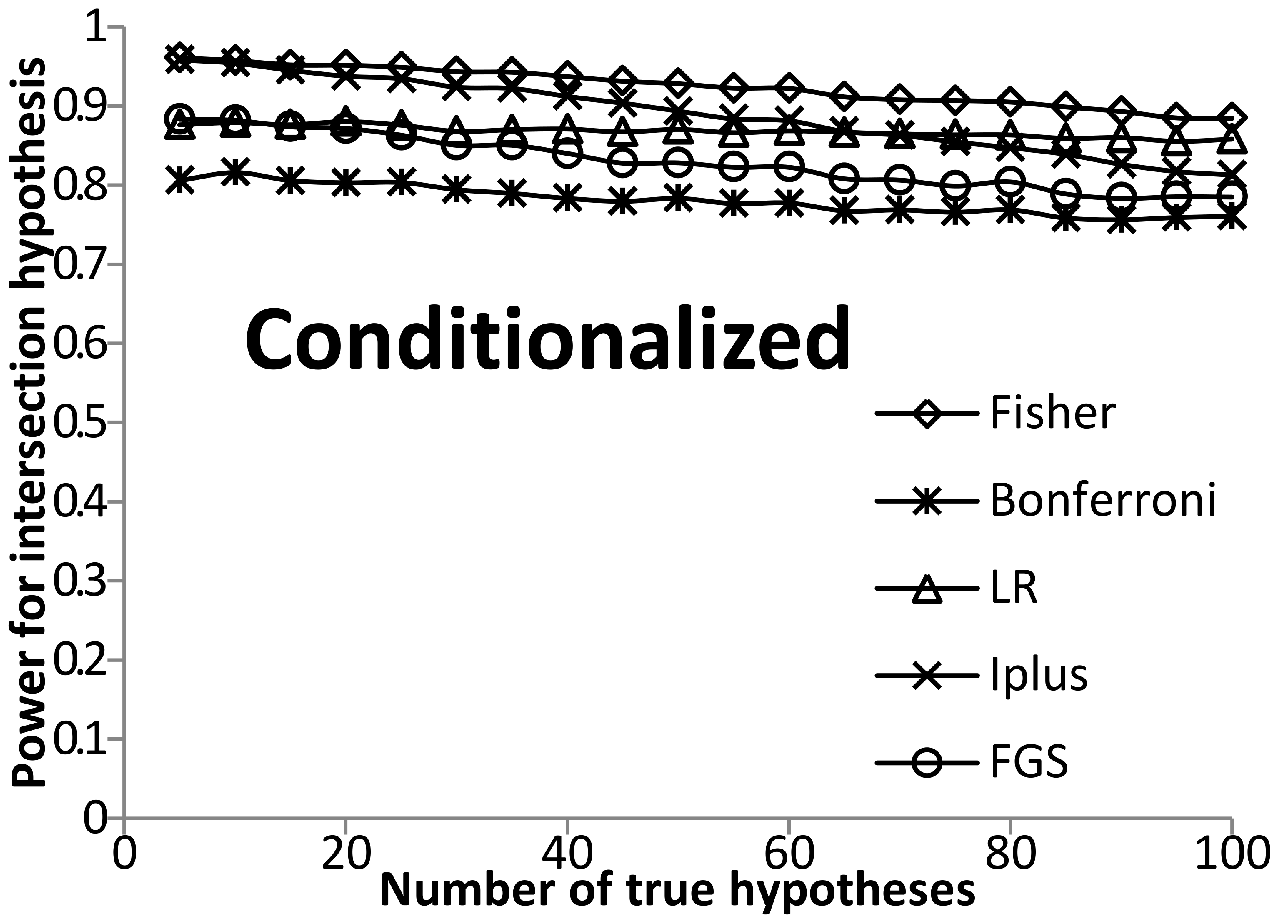}
\caption[Power and m]{Power as a function of the number of true hypothesis. The number of false hypotheses is fixed at 5 in all points. The number of true hypotheses increases from left to right. The considered noncentrality parameter is $\pm 2$.}
\label{fig:figureM}
\end{figure}

\begin{figure}[!p]
\centering
\includegraphics[width=0.48\linewidth]{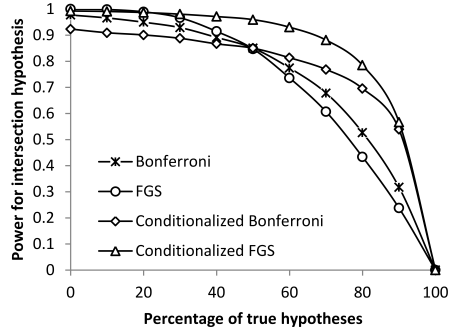}
\caption[Power and percentage]{Power as a function of the percentage of true hypothesis. The total number of hypotheses is fixed at 100 in all points. The percentage of true hypotheses increases from left to right. The considered noncentrality parameter is $\pm 1.5$.}
\label{fig:figureP2}
\end{figure}

\begin{figure}[!p]
\centering
\includegraphics[width=0.48\linewidth]{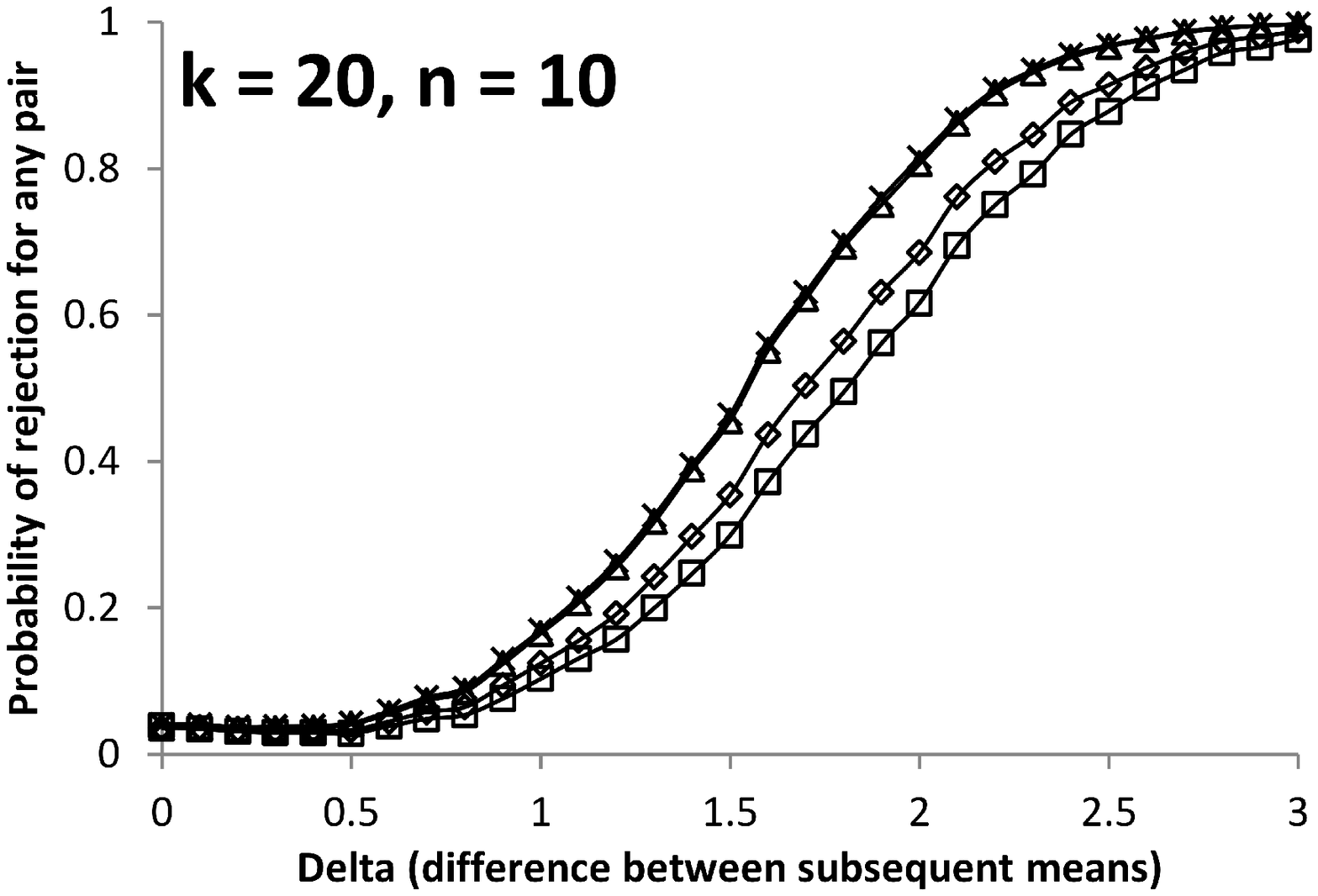}
\includegraphics[width=0.48\linewidth]{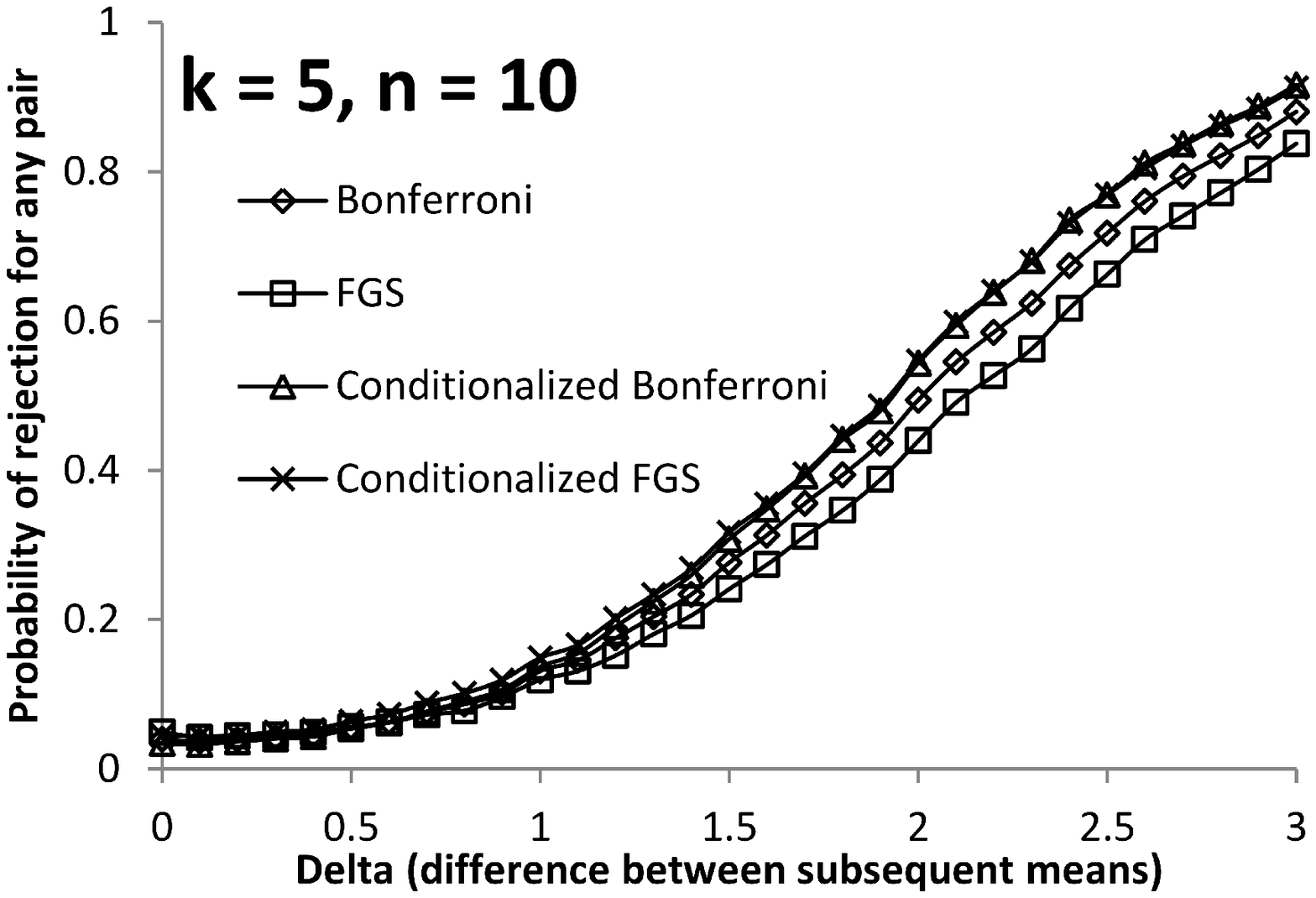}
\caption[Power in pairwise comparisons]{Probability of rejecting any one-sided hypothesis $\mu_i\leq\mu_j$ with $i < j$ if all pairwise comparisons are made. The probability is plotted as a function of the difference $\delta$ between the first $k-1$ consecutive means. The last mean is equal to the first. The number of means was $k=20$ or $k=5$.}
\label{fig:figurePairs}
\end{figure}

\begin{figure}[!p]
\begin{center}
\includegraphics[width=1\linewidth]{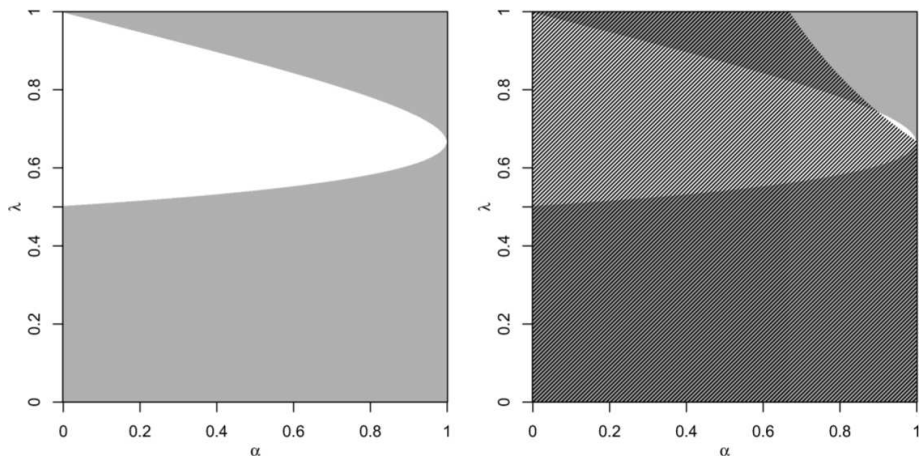}
\end{center}
\caption{Permissible ranges for $\lambda$ given $\alpha$. On the left the grey area shows the permissible combinations of $\lambda$ and $\alpha$ based on Lemma \ref{lem.bivariate_PQD}. On the right the solid grey and dashed areas show the combinations permitted by both Lemmas \ref{lem.bivariate_PQD} and \ref{lem.bivariate_normal}.}
\label{fig.lambda_ok_rough_QPD}
\end{figure}
\end{document}